\documentclass[journal]{IEEEtran}

\usepackage{amsmath,amsthm,amssymb,bm,cite}
\newtheorem{theorem}{Theorem}
\newtheorem{lemma}{Lemma}
\newtheorem{definition}{Definition}

\newtheorem{remark}{Remark}
\usepackage{amsfonts}
\usepackage{amsmath}
\usepackage{amssymb}
\usepackage{amsthm}
\usepackage{array}
\usepackage{caption}
\usepackage{graphicx}
\usepackage{latexsym}
\usepackage{mathtools}
\usepackage{subfigure}
\usepackage{tikz}
\usepackage{times}
\usepackage{tkz-berge}
\usepackage{xparse}
\usepackage{setspace}

\title{On Deterministic Sampling Patterns for Robust Low-Rank Matrix Completion}

\author{Morteza Ashraphijuo, Vaneet Aggarwal, and Xiaodong Wang\thanks{Morteza Ashraphijuo and Xiaodong Wang are with the Department of Electrical Engineering, Columbia University, NY, email: \{ashraphijuo,wangx\}@ee.columbia.edu. Vaneet Aggarwal is with the School of Industrial Engineering, Purdue University, West Lafayette, IN, email: vaneet@purdue.edu. }}

\begin{document}
\maketitle
\begin{abstract}
In this letter, we study the deterministic sampling patterns for the completion of low rank matrix, when corrupted with a sparse noise, also known as robust matrix completion. We extend the recent results on the deterministic sampling patterns in the absence of noise based on the geometric analysis on the Grassmannian manifold. A special case where each column has a certain number of noisy entries is considered, where  our probabilistic analysis performs very efficiently. Furthermore, assuming that the rank of the original matrix is not given, we provide an analysis to determine if the rank of a valid completion is indeed the actual rank of the data corrupted with sparse noise by verifying some conditions.
\end{abstract}

\section{Introduction}

This letter considers the problem of recovering low rank matrix, when corrupted with a sparse noise. This problem is called Robust Matrix Completion. This problem has been studied widely, see for instance  \cite{hsu2011robust,netrapalli2014non,cherapanamjeri2016nearly,chbust,klopp2014robust,chen2013low,yi2016fast}, where probabilistic guarantees for either a convex relaxation based approach or alternating minimization based approach are provided and strong assumptions on the value of the entries are required (like coherence condition). In this letter, we consider the deterministic sampling patterns when the data can be completed with a sparse noise and deterministic and probabilistic guarantees for finite and unique completability are provided.

The study of deterministic sampling patterns is motivated by the results in \cite{charact}, where the authors studied the problem for low rank matrix completion. The necessary and sufficient conditions on the sampling pattern are provided for finite completability in \cite{charact}. Moreover, the sampling probability that ensures finite completability is characterized using the deterministic analysis of the sampling pattern. In this work, we extend these results and analyses on the Grassmannian manifold to the case when the sampled data is corrupted by a sparse noise. 

We further consider the case when each column has certain number of noisy entries and provide bounds for the number of samples required in each column. This result resolves the open question in \cite{cherapanamjeri2016nearly}, where the authors asked if $O(r\log N)$ measurements are enough per column for a $d\times N$ matrix where $d<<N$ and a fraction $O(1/r)$ elements are noisy in each column. We answer the question in positive, further decreasing the number of measurements in each column to $O(\max(r,\log d))$. The main idea is to consider all possibilities of the noise support and make use of the existing fundamental conditions on the sampling pattern for the noiseless scenario.


In many situations, the rank of the sampled matrix is unknown, and depending on the data and the sampled entries, there may be rank-$r$ matrices that agree with the observed entries, even if data is not rank-$r$. Thus, guaranteeing whether if there exist a rank $r$ completion of the data, the rank of original data is indeed $r$ has been studied in  \cite{converse,ashraphijuo5}. In this paper, we will  generalize this approach and results to estimate the rank of the sampled matrix corrupted with a sparse noise.

The rest of the letter is organized as follows. Section II describes the model of robust low-rank matrix completion. Section III gives the deterministic conditions on the sampling patterns when the data has infinite, finite, or unique completions in the presence of sparse noise. These results are then specialized in Section IV to the case when each column of the matrix has at-most $g$ noisy entries. Further, the result is extended to give probabilistic guarantees solving the open problem in \cite{cherapanamjeri2016nearly}. Section V gives conditions to determine whether the rank of matrix is indeed $r$ if there exists a valid completion (which mismatches the observed entries on at most the given support) of rank $r$. Some numerical results are provided in Section \ref{simusec}. Finally, Section VI concludes this paper.

\section{Model and Notations}

Suppose we have a rank $r$ data matrix ${\bf X} \in \mathbb{R}^{d\times N}$ having rank $r$. Suppose the data has an added noise $ {\bf W} \in \mathbb{R}^{d\times N}$ such that $ ||{\bf W}||_0\le s$, where   $ ||{\bf W}||_0$ indicates the number of non-zero entries in $ {\bf W}$. Let ${\bf \Omega} \in \mathbb{R}^{d\times N}$ be a binary matrix which indicates the data points where the data is observed. Let ${\bf A}_{\bf \Omega}$ for  given matrices ${\bf A}$ and  ${\bf \Omega}$ (where  ${\bf \Omega}$ is binary) be the matrix with the elements of ${\bf A}$ corresponding to the entries where  ${\bf \Omega}$ has entry 1, and is zero otherwise. The problem for robust matrix completion is to find the rank $r$ matrix ${\bf X}$ when ${\bf \Omega}$  and $({\bf{ X+W}})_{\bf \Omega}$.

Let $m({\bf A})$ denote the number of rows in ${\bf A}$ and $n({\bf A})$ denote the number of columns in ${\bf A}$. Further, let ${\bf C}(\bf \Omega)$ be a modified matrix from a binary matrix ${\bf \Omega}$ as below. 

Consider the $i$-th column of ${\bf \Omega}$ with $l_i$ sampled entries. We construct $l_i - r$ columns (correspond to the $i$-th column of ${\bf \Omega}$) with binary entries such that each column has exactly $r+1$ entries equal to one. Specifically, assume that $x_1, \dots, x_{l_i}$ are the row indices of all observed entries in this column. Let ${\bf C ( \Omega)}_i$ be the corresponding $d \times (l_i - r)$ matrix to this column which is defined such that for any $j \in \{1,\dots,l_i-r\}$, the $j$-th column has the value $1$ in rows $\{x_1,\dots , x_{r},x_{r+j}\}$ and zeros elsewhere. Finally, define the binary matrix ${\bf C( \Omega)} = \left[{\bf C( \Omega)}_1|{\bf C(\Omega)}_2\dots|{\bf C( \Omega)}_{N} \right]$.


We next define the notion of proper submatrix of ${\bf C( \Omega)}$. 

	\begin{definition}
		A submatrix of ${\bf C( \Omega)}$ is called a proper submatrix if its columns correspond to different columns of the sampling pattern $\mathbf{\Omega}$.
	\end{definition}

In Section IV, we also consider the case that each column has almost $g$ noisy elements. In other words, each column of ${\bf W}$ has L-0 norm less than or equal to $g$. 

\section{Sampling Conditions for Matrix Completion with Noisy Entries}

In this section, we will provide the deterministic conditions on the sampling patterns that determine finite or unique completions, when the observed data is corrupted by a sparse noise.  The following lemma is Theorem 1 in \cite{charact}, which provides the necessary and sufficient combinatorial condition on the sampling for finite completability of the matrix $\mathbf{X}$, where it is assumed to be noiseless, i.e., $s=0$.

\begin{lemma}\label{finthdaniel}
Suppose that the matrix is noiseless, i.e., $s=0$.  Assume that each column of ${\bf{ \Omega}}$ has at least $r$ entries which are $1$. For almost every ${\bf X}$ , there exist at most finitely many rank-$r$ completions of ${\bf X}$ if  and only if the following holds. There exists a proper submatrix $\breve{\bf{ \Omega}}$ formed with $r(d-r)$ columns of ${\bf C}(\breve{\bf{ \Omega}})$ such that every matrix ${\bf{\breve \Omega}}'$ formed with a subset of columns in  $\breve{\bf{ \Omega}}$ satisfies
\begin{equation}\label{fieqdan}
m({\bf{\breve \Omega}}') \ge n({\bf{\breve \Omega}}')/r+r.
\end{equation}
\end{lemma}

The following theorem characterizes the conditions on sampling patterns, which results in finite completability for arbitrary values of $s$ (corrupted with a sparse noise). The main idea is to consider all possibilities of the noise support and make use of the existing fundamental conditions on the sampling pattern for the noiseless scenario.

\begin{theorem}[Deterministic Finite Completions]\label{finth}
Assume that each column of ${\bf{ \Omega}}$ has at least $r+s$ entries which are $1$. For almost every ${\bf X}$  and $ {\bf W}$, there exist at most finitely many rank-$r$ completions of ${\bf X}$ if  the following holds. For each $\widehat{\bf{ \Omega}}$ such that $||\widehat{\bf{ \Omega}}||_0 = ||{\bf{ \Omega}}|| -s $ and the entry of $\widehat{\bf{ \Omega}}$ is zero if the corresponding entry of ${\bf{ \Omega}}$ is zero, there exists a proper submatrix $\breve{\bf{ \Omega}}$ formed with $r(d-r)$ columns of ${\bf C}(\widehat{\bf{ \Omega}})$ such that every matrix ${\bf{ \Omega}}'$ formed with a subset of columns in  $\breve{\bf{ \Omega}}$ satisfies
\begin{equation}\label{fieq}
m({\bf{ \Omega}}') \ge n({\bf{ \Omega}}')/r+r, 
\end{equation}
\end{theorem}
\begin{proof}
Recall that in our model there exists at most $s$ noisy observed entries among all sampled entries (non-zero entries of $\mathbf{\Omega}$). Hence, there exists $\widehat{\bf{ \Omega}}$ such that $||\widehat{\bf{ \Omega}}||_0 = ||{\bf{ \Omega}}|| -s $ and all sampled entries corresponding to the entries of $\widehat{\bf{ \Omega}}$ are noiseless. Moreover, according to the assumption of theorem, there exists a matrix $\breve{\bf{ \Omega}}$ formed with $r(d-r)$ columns of ${\bf C}(\widehat{\bf{ \Omega}})$ such that every matrix ${\bf{ \Omega}}'$ formed with a subset of columns in  $\breve{\bf{ \Omega}}$ satisfies \eqref{fieq}. Then, according to Lemma \ref{finthdaniel}, $\mathbf{X}$ is finitely many completable with probability one.
\end{proof}

\begin{remark}
The converse statement of Theorem \ref{finth} holds probabilistically and not necessarily deterministically anymore (with probability one). Because given that for some  $\widehat{\bf \Omega}$, \eqref{fieq} does not hold, with some probability the noise is at the $s$ entries where $\bf \Omega$ is $1$ and $\widehat{\bf \Omega}$ is $0$. This probability depends on the location of the nonzero entries of $\mathbf{W}$. In general there are ${dN}\choose{s}$ possibilities for the location of the noisy entries. So, the converse statement holds true with some probability between $0$ and $1$ depending on the location of noisy entries.
\end{remark}

The following lemma is Theorem 2 in \cite{charact}, which provides the sufficient combinatorial condition on the sampling for unique completability of the matrix $\mathbf{X}$, where it is assumed to be noiseless, i.e., $s=0$.

\begin{lemma}\label{uniqthdaniel}
Suppose that the matrix is noiseless, i.e., $s=0$.  Assume that each column of ${\bf{ \Omega}}$ has at least $r$ entries which are $1$. For almost every ${\bf X}$ , there exist at most finitely many rank-$r$ completions of ${\bf X}$ if  and only if the following holds. There exist disjoint proper submatrices $\breve{\bf{ \Omega}}$ and $\breve{\bf{ \Omega}}_1$ formed with $r(d-r)$ and $(d-r)$ columns of ${\bf C}(\breve{\bf{ \Omega}})$, respectively, such that every matrix ${\bf{\breve \Omega}}'$ formed with a subset of columns in  $\breve{\bf{ \Omega}}$ satisfies
\begin{equation}\label{fieqdanuni1}
m({\bf{\breve \Omega}}') \ge n({\bf{\breve \Omega}}')/r+r, 
\end{equation}
and every matrix ${\bf{\breve \Omega}}'_1$ formed with a subset of columns in  $\breve{\bf{ \Omega}}_1$ satisfies
\begin{equation}\label{fieqdanuni2}
m({\bf{\breve \Omega}}'_1) \ge n({\bf{\breve \Omega}}'_1) +r.
\end{equation}
\end{lemma}

We will next show that if $s+1$ entries are removed rather than $s$ entries and the above guarantees hold, then the  support of ${\bf W}$ (or a superset of it if the support of ${\bf W}$  is smaller than $s$) can be obtained.  Having identified the support of ${\bf W}$, we get the conditions of unique completion as follows.

\begin{theorem}[Deterministic Unique Completion]\label{unith}
Assume that each column of ${\bf{ \Omega}}$ has at least $r+s+1$ entries which are $1$. Suppose that for each $\widehat{\bf{ \Omega}}$ such that $||\widehat{\bf{ \Omega}}||_0 = ||{\bf{ \Omega}}|| -(s+1) $ and the entry of $\widehat{\bf{ \Omega}}$ is zero if the corresponding entry of ${\bf{ \Omega}}$ is zero, if ${\bf C}(\widehat{\bf{ \Omega}})$ contains two disjoint proper submatrices: $\breve{\bf{ \Omega}}$ formed with $r(d-r)$ columns and $\bar{\bf{ \Omega}}$ formed with $(d-r)$ columns, such that

(i) every matrix ${\bf{ \Omega}}'$ formed with a subset of columns in  $\breve{\bf{ \Omega}}$ satisfies
\begin{equation}\label{ident}
m({\bf{ \Omega}}') \ge n({\bf{ \Omega}}')/r+r,
\end{equation}
and 
(ii)  every matrix ${\bf{ \Omega}}'$ formed with a subset of columns in  $\bar{\bf{ \Omega}}$ satisfies
\begin{equation}\label{uniadeq}
m({\bf{ \Omega}}') \ge n({\bf{ \Omega}}')+r.
\end{equation}
Then, almost every rank-$r$ matrix ${\bf X}$ can be recovered from noise where the entries of ${\bf W}$ are generically chosen. 
\end{theorem}

\begin{proof}
As the first step of the proof, given condition \eqref{ident}, we provide a simple algorithm, which identifies the support of noisy entries $\mathbf{W}_{\mathbf{\Omega}}$. The algorithm of completion that we use is the same as above, using every  $\underline{\bf{ \Omega}}$ that has $s$ less entries. To see this, first assume that $|| {\bf W}_{\bf \Omega}||_0=s$.  We note that if the  chosen set $\underline{\bf{ \Omega}}$ is  the same as value of ${\bf{ \Omega}}$ where the noise $ {\bf W}$ entries are removed, there are finite completions by \cite{charact}. However, if the above is not the case, the inherent rank of data with the known entries in  $\underline{\bf{ \Omega}}$ is greater than $r$ since the entries of a matrix with rank $r$ were corrupted by generic entries. We note that since the entries have to match for $\underline{\bf{ \Omega}}$, and if we remove one of the noisy entry (entry in ${\bf W}_{\bf \Omega}$ but not in $\underline{\bf{ \Omega}}$) from $\underline{\bf{ \Omega}}$ to obtain $\widehat{\bf{ \Omega}}$, there are at most a finite number of completions fitting the missing entries.
	
Similarly we can show that if  $|| {\bf W}_{\bf \Omega}||_0 = s-i$, there are finitely many completions since the set of $s$ removals in $\underline{\bf{ \Omega}}$ contain non-zero entries of ${\bf W}_{\bf \Omega}$, $i = 1,\dots,s-1$. Hence, for each value of $i$ that $|| {\bf W}_{\bf \Omega}||_0 = s-i$, there exist exactly one possible support of $W$, which we have identified. Note that the finite sum of finite numbers is also a finite number, and therefore we showed the finite completability for $|| {\bf W}_{\bf \Omega}||_0 \leq s$. This finite number of completions will not match the entry at the noisy part with probability 1. Thus, there cannot be any possible completion with a rank-$r$ matrix which matches all entries of  $\underline{\bf{ \Omega}}$. Hence, we can identify the support of the noise $ {\bf W}$, and therefore condition \eqref{uniadeq} in the statement of the Theorem guarantees unique completability following Lemma \ref{uniqthdaniel}.
\end{proof}

We now restate Theorem 3 in \cite{charact} as the following lemma.

\begin{lemma}\label{danthmmatfin}
Suppose that the matrix is noiseless, i.e., $s=0$. Assume $r \leq \frac{d}{6}$ and that each column of the sampled matrix is observed in at least $l$ entries, uniformly at random and independently across entries, where
\begin{eqnarray}\label{promatrixdanpro}
l > \max\left\{12 \ \log \left( \frac{d}{\epsilon} \right) + 12, 2r\right\}. 
\end{eqnarray}
Also, assume that $ r(d-r) \leq N$. Then, with probability at least $1 - \epsilon$, the assumption on the sampling pattern given in Lemma \ref{finthdaniel} holds, i.e., $\mathbf{X}$ is finitely many completable. Moreover, $ (r+1)(d-r) \leq N$ ensures that with probability at least $1 - \epsilon$, $\mathbf{X}$ is uniquely completable.
\end{lemma}

The uniform sampling result can be described as follows.

\begin{theorem}[Probabilistic Finite and Unique Completion]\label{**}
Suppose $r\leq\frac{d}{6}$, and each column includes  at least $l$ observed entries, where
\begin{eqnarray}\notag
l- 12(r+s+1)\log(l/(r+s+1))> \nonumber \\ \max\{12(\log(\frac{d}{\epsilon})+r+s+1),2r,2r+s+1\}.
\end{eqnarray}
Then, with probability at least $1-\epsilon$, almost every ${\bf X}$ will be finitely completable if $N\geq r(d-r)$ and uniquely completable if $N\geq (r+1)(d-r)$.
\end{theorem}
\begin{proof}
This is a simple extension of Lemma \ref{danthmmatfin} by using union bound over all at most $\binom{l}{r+s+1}$ choices for each column since the property has to hold over all such choices of $r+s+1$ removals in each of the columns.
\end{proof}

\section{Sampling Conditions for Completion with Noisy Entries in each Column}

Having $s$ entries in ${\bf W}$  anywhere in the data makes each column to have at least $O(r+s)$ elements which is large for a large-scale matrix. We next consider a structure where each column has almost $g$ noisy elements. In other words, each column of ${\bf W}$ has L-0 norm less than or equal to $g$. Then, the Theorems \ref{finth} and \ref{unith} can be easily extended to consider different pattern on $\widehat{\bf{ \Omega}}$, and for completion the modified Theorems \ref{finth} and \ref{unith} are described as follows.

\begin{theorem}[Deterministic Finite Completion for Column-wise Sparse Noise]\label{finitclth}
Assume that each column of ${\bf{ \Omega}}$ has at least $r+g+1$ entries which are $1$. Suppose that for each $\widehat{\bf{ \Omega}}$ such that each column of $\widehat{\bf{ \Omega}}$ has $g+1$ less ones than that in ${\bf{ \Omega}}  $, and the entry of $\widehat{\bf{ \Omega}}$ is zero if the corresponding entry of ${\bf{ \Omega}}$ is zero, if ${\bf C}(\widehat{\bf{ \Omega}})$ contains a matrix: $\breve{\bf{ \Omega}}$ formed with $r(d-r)$ columns such that
	
	(i) every matrix ${\bf{ \Omega}}'$ formed with a subset of columns in  $\breve{\bf{ \Omega}}$ satisfies
	\begin{equation}
	m({\bf{ \Omega}}') \ge n({\bf{ \Omega}}')/r+r.
	\end{equation}

Then, for almost every rank-$r$ matrix ${\bf X}$, there exist finitely many rank-$r$ completions, where the entries of ${\bf W}$ are generically chosen.
\end{theorem}
 
\begin{theorem}[Deterministic Unique Completion for Column-wise Sparse Noise]\label{uniclth}
Assume that each column of ${\bf{ \Omega}}$ has at least $r+g+1$ entries which are $1$. Suppose that for each $\widehat{\bf{ \Omega}}$ such that each column of $\widehat{\bf{ \Omega}}$ has $g+1$ less ones than that in ${\bf{ \Omega}}  $, and the entry of $\widehat{\bf{ \Omega}}$ is zero if the corresponding entry of ${\bf{ \Omega}}$ is zero, if ${\bf C}(\widehat{\bf{ \Omega}})$ contains two disjoint proper submatrices: $\breve{\bf{ \Omega}}$ formed with $r(d-r)$ columns and $\bar{\bf{ \Omega}}$ formed with $(d-r)$ columns, such that
	
	(i) every matrix ${\bf{ \Omega}}'$ formed with a subset of columns in  $\breve{\bf{ \Omega}}$ satisfies
	\begin{equation}
	m({\bf{ \Omega}}') \ge n({\bf{ \Omega}}')/r+r,
	\end{equation}
	and 
	(ii)  every matrix ${\bf{ \Omega}}'$ formed with a subset of columns in  $\bar{\bf{ \Omega}}$ satisfies
	\begin{equation}
	m({\bf{ \Omega}}') \ge n({\bf{ \Omega}}')+r.
	\end{equation}
	Then almost every rank-$r$ matrix ${\bf X}$ can be recovered from noise where the entries of ${\bf W}$ are generically chosen. 
\end{theorem}

Having identified the sampling conditions on robust data completion, we now determine the uniform random sampling conditions for the data completion. 

\begin{theorem}[Probabilistic Finite and Unique Completion for Column-wise Sparse Noise]\label{*}
Suppose $r\leq\frac{d}{6}$, and each column includes  at least $l$ observed entries, where
\begin{eqnarray}\label{fhkfj}
l- 12(g+1)\log(l/(g+1))> \nonumber \\ \max\{12(\log(\frac{d}{\epsilon})+g+1),2r,r+g+1\}.
\end{eqnarray}
Then, with probability at least $1-\epsilon$, almost every ${\bf X}$ will be finitely completable if $N\geq r(d-r)$ and uniquely completable if $N\geq (r+1)(d-r)$.
\end{theorem}
\begin{proof}
This is a simple extension of Lemma \ref{danthmmatfin} by using union bound over all $\binom{l}{g+1}$ choices in the column since the property has to hold over all such choices of $g+1$ removals in each of the columns.
\end{proof}

We note that for $g\approx l/r$ ($r=\omega(1)$) in each column where $l$ entries are observed, the number of samples needed in each column is $O(\max(r,\log(d)))$. We note that this setting proposed an an open problem in \cite{cherapanamjeri2016nearly}, where the authors asked if the observations of $O(r\log(d))$ per column suffice. In this paper, we answer this question in positive, and further reducing the number of observations needed to $O(\max(r,\log(d)))$. This result does not necessarily need $g\approx l/r$, while will work as long as $g/l = o(1)$ or $l>>g$.


\section{Sampling Conditions for Rank Estimation with noisy entries}

So far, we assumed that the rank $r$ of the matrix is known. In this section, we assume that the value of the rank, $r$, is not given and we are interested in approximating it. The following lemma is restatement of Corollary 1 in \cite{ashraphijuo5}.

\begin{lemma}\label{rankdetlem}
Suppose that the matrix is noiseless, i.e., $s=0$. Define $\mathcal{S}_{\mathbf{\Omega}}=\{1,2,\dots,r^*\}$, where $r^*$ is the maximum number such that the assumption on the sampling pattern given in Lemma \ref{finthdaniel} holds true, i.e., $r^*$ is the maximum number such that there are finitely many completions of $\mathbf{X}$ of rank $r^*$, and let $r^{\prime} \in \mathcal{S}_{\mathbf{\Omega}}$. Then, with probability one, exactly one of the followings holds

(i) $r \in \{1,2,\dots,r^{\prime}\}${\rm ;}

(ii) For any arbitrary completion of the matrix $\mathbf{X}$ of rank $r^{\prime \prime}$, we have $r^{\prime \prime} \notin \{1,2,\dots,r^{\prime}\}$.
\end{lemma}

The following theorem extends the above lemma to the case of existence of sparse noise over the entire data.

\begin{theorem}[Deterministic Conditions for Rank Estimation for Robust Completion]\label{rankdetnoise}
Define $\mathcal{S}_{\mathbf{\Omega}}=\{1,2,\dots,r^*\}$, where $r^*$ is the maximum number such that the assumption on the sampling pattern given in Theorem \ref{finth} holds true and let $r^{\prime} \in \mathcal{S}_{\mathbf{\Omega}}$. Then, with probability one, exactly one of the followings holds

(i) $r \in \{1,2,\dots,r^{\prime}\}${\rm ;}

(ii) For any arbitrary completion of the matrix $\mathbf{X}$ of rank $r^{\prime \prime}$, we have $r^{\prime \prime} \notin \{1,2,\dots,r^{\prime}\}$.
\end{theorem}

\begin{proof}
According to Theorem \ref{finth}, for any $r^{\prime} \in \mathcal{S}_{\mathbf{\Omega}}$, there exist finitely many completions of $\mathbf{X}$ of rank $r^{\prime}$. The rest of the proof follows from Lemma \ref{rankdetlem}.
\end{proof}

\begin{theorem}[Probabilistic Conditions for Rank Estimation for Robust Completion]\label{rankrprobnoise}
Suppose $r\leq\frac{d}{6}$, $r(d-r) \leq N$ and let $r^{\prime} \in \mathcal{S}_{\mathbf{\Omega}}$ such that each column includes  at least $l$ observed entries, where
\begin{eqnarray}
l- 12(g+1)\log(l/(g+1))> \nonumber \\ \max\{12(\log(\frac{d}{\epsilon})+g+1),2r^{\prime},r^{\prime}+g+1\}.
\end{eqnarray}

Then, with probability one, exactly one of the followings holds

(i) $r \in \{1,2,\dots,r^{\prime}\}${\rm ;}

(ii) For any arbitrary completion of the matrix $\mathbf{X}$ of rank $r^{\prime \prime}$, we have $r^{\prime \prime} \notin \{1,2,\dots,r^{\prime}\}$.
\end{theorem}

\begin{proof}
Define $\mathcal{S}_{\mathbf{\Omega}}=\{1,2,\dots,r^*\}$, where $r^*$ is the maximum number such that the assumption on the sampling pattern given in Theorem \ref{finth} holds true. According to Theorem \ref{**}, there exist finitely many completions of $\mathbf{X}$. Hence, $r^{\prime} \leq r^*$, and therefore $r^{\prime} \in \mathcal{S}_{\mathbf{\Omega}}$. The rest of the proof follows from Theorem \ref{rankdetnoise}.
\end{proof}

\begin{remark}
Theorems \ref{rankdetnoise} and \ref{rankrprobnoise} can be directly written for noisy entries in each column, where assumption on the sampling pattern given in Theorems \ref{finth} and \ref{**} are replaced by the assumption on the sampling pattern given in Theorems \ref{finitclth} and \ref{*}, respectively.  
\end{remark}

\begin{remark}
Define $\mathcal{S}_{\mathbf{\Omega}}=\{1,2,\dots,r^*\}$, where $r^*$ is the maximum number such that the assumption on the sampling pattern given in Theorem \ref{finth} holds true. Assume that there exist a completion of the matrix $\mathbf{X}$ of rank $r^{\prime} \in \mathcal{S}_{\mathbf{\Omega}}$. Then, according to Theorem \ref{rankdetnoise}, with probability one, $r \leq r^{\prime}$. 
\end{remark}

\section{Numerical Results}\label{simusec}

In this section, we consider $\mathbf{X} \in \mathbb{R}^{600 \times 60000}$ and change the value of rank from $1$ to $100$ and $\epsilon = 0.01$. We consider the uniform number of noisy entries in each column and compare the bounds given in \eqref{promatrixdanpro} (noiseless) and \eqref{fhkfj} (noisy) for $g=1$ and $g=2$. For example, $g=1$ means that each column has one noisy entry and $60000$ noisy entries in total. The corresponding bounds result in different portions of the samples, which are shown in Figure \ref{fig2}. 

\begin{figure}[h]
	\centering
		{\includegraphics[width=8cm]{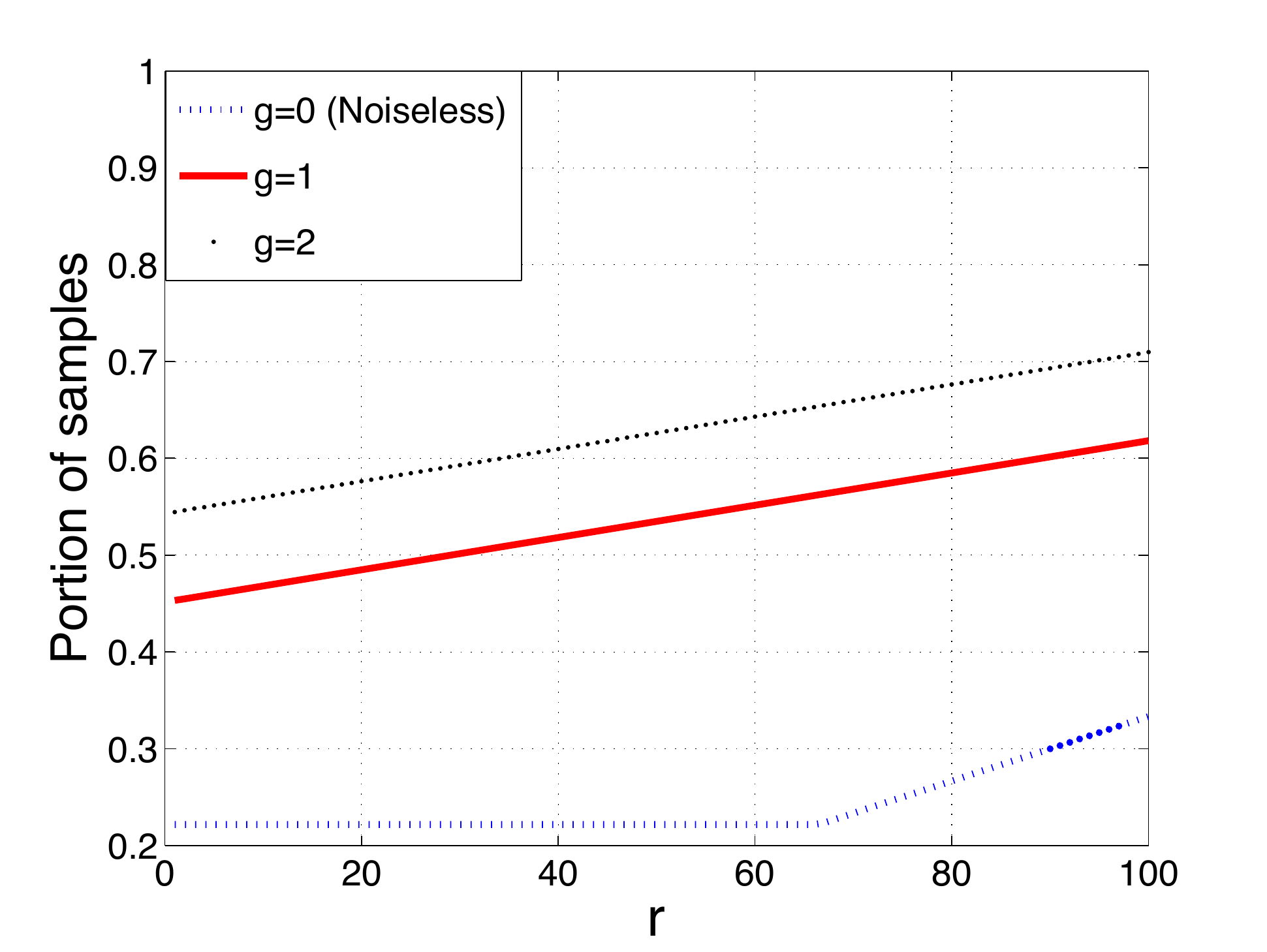}}
	\caption{ \footnotesize Comparison of portion of the required number of samples for finite/unique completability for different values of the number of noisy entries at each column, i.e., $g$.}
	\label{fig2}\vspace{-4mm}
\end{figure}


\section{Conclusions}
We studied the conditions on the sampling patterns for the completion of low rank matrix, when corrupted with a sparse noise. Both general sparse noise in the matrix and column-wise sparse noise models are considered. Using these results, an open question in \cite{cherapanamjeri2016nearly} is resolved with improved results. Furthermore, assuming that the rank of the original matrix is not given, we provide an analysis to verify if the rank of a given valid completion is indeed the actual rank of the matrix. The approach in this paper can be easily extended to other tensor structures like Tucker rank, Tensor-train rank, CP-rank, multi-view data, since the corresponding results without noise are given in
\cite{ashcon1,ashraphijuo4,ashcon2,ashraphijuo3,ashraphijuo5,ashraphijuo,ashraphijuo2}. 

Finding computationally efficient algorithms that achieve close to these bounds is an open problem. Some of the existing algorithms use alternating minimization based approaches \cite{cherapanamjeri2016nearly,yi2016fast}, which could also be extended to tensors following approaches in \cite{Wang_2017_ICCV,liu9848low,liu2016low}.

\section*{Acknowledgment}
This work was supported in part by the U.S. National Science Foundation (NSF) under grant CIF1064575, and in part by the U.S. Office of Naval Research (ONR) under grant N000141410667. The authors would like to thank Ruijiu Mao of Purdue University for helpful discussions.

\bibliographystyle{IEEETran}
\bibliography{bib}

\end{document}